\documentclass[%
 reprint,
superscriptaddress,
nofootinbib,
 amsmath,amssymb,
 aps,
]{revtex4-2}

\usepackage{graphicx}% Include figure files
\usepackage{dcolumn}% Align table columns on decimal point
\usepackage{bm}% bold math
\usepackage{braket}
\usepackage{amsthm}
\usepackage{xcolor}
\usepackage[colorlinks=true,
            linkcolor=blue,
            citecolor=blue,
            urlcolor=blue,
            breaklinks=true]{hyperref}% hypertext capabilities
\newtheorem{theorem}{Theorem}
\newtheorem{lemma}{Lemma}

\def\bea{\begin{eqnarray}}
\def\eea{\end{eqnarray}}
\def\<{\langle}
\def\>{\rangle}
\def\be{\begin{equation}}
\def\ee{\end{equation}}
\newcommand{\Tr}{\operatorname{Tr}}
\newcommand{\HA}{\mathcal{H}_A}
\newcommand{\HB}{\mathcal{H}_B}
\newcommand{\ketbra}[2]{|#1\rangle\langle #2|}

\begin{document}

\preprint{APS/123-QED}

\title{Universal Entanglement Dynamics of Unitary Operators}

\author{Ian Low}
 \affiliation{%
Department of Physics and Astronomy, Northwestern University, Evanston, IL60208, USA}%
\affiliation{%
High Energy Physics Division, Argonne National Laboratory, Argonne, IL 60439, U.S.A.}%

\author{Navin McGinnis}
\affiliation{%
High Energy Physics Division, Argonne National Laboratory, Argonne, IL 60439, U.S.A.}%
 \affiliation{Enrico Fermi Institute, Physics Department, University of Chicago, Chicago, IL 60637, USA}
\affiliation{%
 Department of Physics, University of Arizona, Tucson, Arizona 85721, USA}%

\date{\today}% It is always \today, today,

\begin{abstract}

The entangling power of a unitary operator acting on a bipartite Hilbert
space measures the entanglement it generates from product states,
averaged over the inputs. A finite-dimensional unitary has a spectral
decomposition $U=\sum_{a=1}^{n}e^{i\theta_a}P_a$, where $e^{i\theta_a}$ are the eigenvalues, $P_a$ the corresponding eigen-projectors, and $n$ is the
number of distinct eigenvalues. After removing an overall phase, the
entangling power is  a function on the $(n-1)$-torus of
relative eigenphases at fixed spectral projectors.  We prove that this function is stationary at all
$2^{n-1}$ points on the torus where every relative phase is $0$ or
$\pi$, which we define as \textit{corners}. Up to an overall phase,  $U$ at each corner  is a
generalized reflection $R=\mathbb{I}-2Q$ satisfying $R^2=\mathbb{I}$,
where $Q$ is the sum of spectral projectors whose relative phase is $\pi$. At the corner the  entangling
power  is expressed in terms of seven local-unitary invariants of $Q$. A unitary gate $U$ can be realized as a
corner  of some projector family if and only if
$U^2\propto\mathbb{I}$, a condition satisfied by many Clifford and non-Clifford  gates.  We illustrate the theorem with two-qubit gates, $SU(N)$ channel decompositions, and two-site spin chains, obtaining examples of minima, maxima, and saddle points. In addition, a corner that is a saddle point on the full phase torus can appear as a local maximum or minimum along different time-evolution trajectories.
\end{abstract}

\maketitle

%\tableofcontents

\section{\label{sec:intro} Introduction}

The entangling power of a unitary operator quantifies the entanglement it generates from
product inputs, averaged over those inputs~\cite{Zanardi:2000zz,Zanardi:2001zza}. It is a
property of the operator rather than of any particular state, and it presupposes no
knowledge of the symmetries of the system. Applied to a time-evolution operator and
averaged over time, it characterizes the typical entanglement production of a
dynamics~\cite{Lu:2008apn,Pal:2018ptv,Manna:2024jvb}.

A growing body of work has found that the entanglement generated in particle scattering is
extremized at points of enhanced symmetry. For instance, it has been showne that the nucleon--nucleon $S$-matrix is minimally
entangling at the Wigner $SU(4)$ symmetric point~\cite{Beane:2018oxh,Low:2021ufv,Liu:2022grf},
with similar correlations in low-energy hadron
dynamics~\cite{Liu:2023bnr,Hu:2025lua,Low:2026evp}, extensions of the Higgs
sector and the Standard Model~\cite{Carena:2023vjc,Carena:2025wyh,Li:2026kha,Kanemura:2026won}, and in the general relativistic $S$-matrix framework~\cite{Cervera-Lierta:2017tdt,McGinnis:2025brt,McGinnis:2025iab,McGinnis:2025xgt}. These
observations have motivated the broader possibility that symmetry is tied to an
entanglement-extremization principle, and that the extremization of quantum resources may
help organize fundamental parameters, from the flavor structure~\cite{Thaler:2024anb} to the Higgs mass~\cite{Liu:2025iwh}. A closely related line of work replaces entanglement by
nonstabilizerness, or ``magic,'' and asks whether that resource organizes physical
structure as
well~\cite{Liu:2025bgw,Cao:2026aye,Liu:2025qfl,Liu:2025frx,Li:2026udy,Busoni:2025dns,Gargalionis:2025iqs,Li:2026zum,Chang:2026nzq}.
The same correlation has been seen in other contexts, in the time-averaged entangling
power of anisotropic Heisenberg spin chains~\cite{Low:2026oyf} and the quantum Rabi
models~\cite{Low:2026kxb}.

In these studies the extremum of the entangling power is always computed on a case-by-case basis: the entangling power
is computed for a particular system, and one then reads off the  stationary points, either
numerically or by inspection of an explicit closed form. What is missing is a general statement about where such extrema exist.

In this work we supply a theorem showing that a discrete set of such stationary points 
occurs at universal values of the relative eigenphases, independently of the eigen-projectors and of the interaction that generates the phases. More specifically, for a unitary operator  $U$ on a finite-dimensional space with $n$ distinct eigenvalues, one can write the following spectral decomposition,
\be
U=\sum_{a=1}^{n} e^{i\theta_{a}}P_{a} = e^{i\theta_n}\left(\sum_{a=1}^{n-1} e^{i\delta_a} P_a + P_n\right)\ ,
\label{eq:spectral}
\ee
where $e^{i\theta_a}$ is the eigenvalue, $P_a$  the corresponding eigen-projector,  and $\delta_a\equiv \theta_a -\theta_n,$  the relative eigenphase. Without loss of generality we have factored out the overall phase $e^{i\theta_n}$ in the above. The entangling
power is then a function of the eigen-projectors $P_{a}$ and of the $(n-1)$-torus of relative eigenphases $\delta_{a}$. The theorem states that the entangling power is stationary
at the $2^{n-1}$ points of that torus where $\delta_{a}\in \{0,\pi\}$, which we call the \textit{corners}. The
statement holds for any projectors, any subsystem dimensions, and any choice of bipartition.

At a corner the gate takes the form of a generalized reflection $R=\mathbb{I}-2Q$, $R^2=\mathbb{I}$, where $Q$
is the sum of the projectors whose relative phase is $\pi$.
We obtain its entangling power in closed form,
and for the case of two projectors, $n=2$, the resulting expression solves the landscape completely and supplies a
criterion for whether the nontrivial corner is a maximum or a minimum. Many standard quantum gates arise as corners of two-projector families: not only the Clifford gates ${\rm CNOT}$, ${\rm SWAP}$ and ${\rm CZ}$, but also the non-Clifford Toffoli and controlled-controlled-$Z$ gates. The corners are  universal in terms of the relative eigenphases, whatever the projectors: two unitary operators with the same $\delta_a$ but different $P_a$ have stationary entangling power at the same places, $\delta_a\in\{0,\pi\}$.

It is worth emphasizing that the theorem specifies the location of the 
stationary points with respect to the relative eigenphases but does not classify them; whether a corner is a minimum, a maximum or
a saddle point requires computing the second derivative,  the Hessian matrix. In Sec.~\ref{sec:apps} we exhibit all three possibilities. The corners are also not the only stationary points; the others are not covered by the theorem.

This paper is organized as follows. In Sec.~\ref{sec:setup} we review
the entangling-power formalism and establish the invariances used
below. In Sec.~\ref{sec:thm}, we provide the statement and proof of the theorem described above, derive the
entangling power at an arbitrary corner, solve the general two-projector
entanglement landscape, and classify the standard corner gates. In
Sec.~\ref{sec:apps} we apply the results to two-qubit dynamics,
two-projector channels with an $SU(N)$ symmetry, and the two-site spin-$1$
bilinear--biquadratic model. We summarize the
scope of the theorem and discuss open questions in
Sec.~\ref{sec:con}.

%%%%%%%%%%%%%%%%%%%%%%%%%%%%%%%%%%%%%%%%%%%%%%%%%%%%%%%%%%%%%%%%%%%%%%%%
\section{\label{sec:setup} Entangling Power}
%%%%%%%%%%%%%%%%%%%%%%%%%%%%%%%%%%%%%%%%%%%%%%
Consider unitary operators acting on a bipartite Hilbert space $\HA\otimes\HB$, with
$\HA\cong\mathbb{C}^{d_{A}}$ and $\HB\cong\mathbb{C}^{d_{B}}$ finite dimensional. We will use Latin
indices $i,j$ to label an orthonormal basis of $\HA$ and Greek indices $\alpha,\beta$ for one of
$\HB$. We measure the entanglement of a state in $\HA\otimes\HB$ by the linear entropy
$E(\ket{\psi})=1-\Tr\rho_{A}^{2}$, with $\rho_{A}=\Tr_{B}\ketbra{\psi}{\psi}$, and define
the \textit{entangling power} of a unitary operator $U$ as the average of $E$ over its action on product
inputs~\cite{Zanardi:2000zz,Zanardi:2001zza},
\begin{equation}
  \mathcal{E}_{p}(U)\equiv\int d\mu_{A}\,d\mu_{B}\;
  E\big(U\ket{\phi_{A}}\otimes\ket{\phi_{B}}\big),
  \label{eq:EP_def}
\end{equation}
where $d\mu_{A}$ and $d\mu_{B}$ are Haar measures on $\HA$ and $\HB$, respectively.\footnote{It is possible to adopt a different sampling procedure from the Haar sampling. See Ref.~\cite{Low:2026kxb}.} Notice that $\mathcal{E}_{p}$ is invariant under local unitaries,
$\mathcal{E}_{p}(u_{1}\otimes u_{2}\cdot U\cdot u_{3}\otimes u_{4})=\mathcal{E}_{p}(U)$,
a property used throughout this work.

The average in Eq.~\eqref{eq:EP_def} can be computed in closed form. To see this we consider two
rearrangements of the matrix elements $U_{(i\alpha),(j\beta)}$~\cite{Lu:2008apn}. One is the partial
transpose on $\HA$,
\begin{equation}
  \big(U^{T_A}\big)_{(i\alpha),(j\beta)} \equiv U_{(j\alpha),(i\beta)},
  \label{eq:PTdef}
\end{equation}
which remains a $d_{A}d_{B}\times d_{A}d_{B}$ matrix, while the other is the realignment which instead groups the
two row indices against the two column indices,
\begin{equation}
  \big(U^{R}\big)_{(ij),(\alpha\beta)} \equiv U_{(i\alpha),(j\beta)},
  \label{eq:Rdef}
\end{equation}
and is a $d_{A}^{2}\times d_{B}^{2}$ matrix. Neither is unitary, even for unitary $U$.
From them we form the two invariants
\begin{equation}
  I_R(U) \equiv \Tr\!\big[(U^{R}U^{R\dagger})^2\big],
  \quad
  I_T(U) \equiv \Tr\!\big[(U^{T_A}U^{T_A\dagger})^2\big],
  \label{eq:invariants}
\end{equation}
in terms of which the entangling power is given by~\cite{Ma:2007tny,Lu:2008apn}
\begin{equation}
  \mathcal{E}_p(U) \;=\; \frac{d_A d_B+1}{(d_A+1)(d_B+1)}
             \;-\; \frac{I_R(U)+I_T(U)}{d_A(d_A+1)\,d_B(d_B+1)} .
    \label{eq:master}
\end{equation}
Equation~\eqref{eq:master} is the main computational tool we need. Its derivation, which
rests on the two-copy Haar average and on writing the purity as a SWAP on the doubled
Hilbert space, is given in Refs.~\cite{Ma:2007tny,Lu:2008apn}.
% we refer the reader to
%Refs.~\cite{Low:2026oyf,Low:2026kxb} for a detailed account in the present notation, noting that those works normalize the entangling power by an additional factor $d_{A}/(d_{A}-1)$. 
%Two special cases fix the conventions. For $U=\mathbb{I}$ one finds
%$I_{R}=d_{A}^{2}d_{B}^{2}$ and $I_{T}=d_{A}d_{B}$, and for $d_{A}=d_{B}$ and $U=S_{W}$
%the SWAP, $S_{W}\ket{x}\otimes\ket{v}=\ket{v}\otimes\ket{x}$, the two values are
%exchanged; in both cases $\mathcal{E}_{p}=0$, as it must be for operators that map every
%product state to a product state.

The theorem of Sec.~\ref{sec:thm} requires one further property: the entangling power
is the same for a unitary and its Hermitian conjugate. The
corresponding statement for the operator entanglement was given by Zanardi~\cite{Zanardi:2001zza}, and the two-qubit case is recorded in Refs.~\cite{VatanWilliams2004,Rezakhani2004}. Here we state it in arbitrary
dimensions, with a short proof, since it is the only ingredient the theorem requires.
\begin{lemma}\label{lem:dagger}
For every unitary $U$ acting on $\HA\otimes\HB$ and every phase $\phi$,
\begin{equation}
  \mathcal{E}_p\!\big(e^{i\phi}U^{\dagger}\big) = \mathcal{E}_p(U).
  \label{eq:daggerlemma}
\end{equation}
\end{lemma}
\begin{proof}
It is immediate from Eq.~(\ref{eq:invariants}) that the overall phase cancels. For the
adjoint, Eq.~(\ref{eq:PTdef}) gives $(U^{\dagger})^{T_{A}}=(U^{T_{A}})^{\dagger}$, so that
$I_{T}(U^{\dagger})=I_{T}(U)$ by the cyclic property of the trace in
Eq.~\eqref{eq:invariants}. For the realignment the adjoint cannot enter in the same way,
since $(U^{R})^{\dagger}$ is not even of the same shape as $(U^\dagger)^R$ unless $d_{A}=d_{B}$; instead
recall that the singular values of $U^{R}$ are the operator-Schmidt
coefficients of $U$. Writing the operator-Schmidt decomposition
$U=\sum_{k}s_{k}\,A_{k}\otimes B_{k}$, with $\{A_{k}\}$ and $\{B_{k}\}$ orthonormal in the
Hilbert-Schmidt inner product and $s_{k}\geq0$, one has $I_{R}(U)=\sum_{k}s_{k}^{4}$. Taking the
adjoint gives $U^{\dagger}=\sum_{k}s_{k}\,A_{k}^{\dagger}\otimes B_{k}^{\dagger}$, and the sets
$\{A_{k}^{\dagger}\}$ and $\{B_{k}^{\dagger}\}$ are again orthonormal. Since any such expansion with non-negative coefficients is a singular-value decomposition of the corresponding realignment, the two decompositions have the same singular values, and $I_R(U^{\dagger})=I_R(U)$.
The Lemma follows.
\end{proof}

A second invariance is needed below. For $d_A=d_B$ the entangling power is also unchanged
by composition with $S_W$,
$\mathcal{E}_p(S_W U)=\mathcal{E}_p(U S_W)=\mathcal{E}_p(U)$, a property established
already in Ref.~\cite{Zanardi:2000zz}. In the variables used here it is the statement
that composing with $S_W$ exchanges the two invariants, $I_R(U S_W)=I_T(U)$ and
$I_T(U S_W)=I_R(U)$, so that their sum, the only combination entering
Eq.~\eqref{eq:master}, is unaffected. %see also Refs.~\cite{Jonnadula2017,Jonnadula2020},
%where the antisymmetric combination $I_R-I_T$ is used to define the complementary
%quantity known as gate typicality.
%%%%%%%%%%%%%%%%%%%%%%%%%%%%%%%%%%
%%%%%%%%%%%%%%%%%%%%%%%%%%%%%%%%%%%%%%%%%%%%%%%%%%%%%%%%%%%%%%%%%%%%%%%%
\section{\label{sec:thm} A Theorem on Stationary Points}
For a unitary $U$ acting on
$\mathcal{H}=\mathcal{H}_A\otimes\mathcal{H}_B$, we consider its form as written in Eq.~\eqref{eq:spectral} in
terms of its eigenphases $\theta_{a}$ and eigen-projectors $P_{a}$, which satisfy
$P_{a}P_{b}=\delta_{ab}P_{a}$ and $\sum_{a}P_{a}=\mathbb{I}$.
In finite dimensions every unitary admits that form.

% so it is no restriction on $U$. Fixing the resolution $\{P_{a}\}$ and letting the eigenphases run organizes a \textit{family} of unitaries, as for a partial-wave $S$-matrix as a function of energy or a calibrated phase gate as a function of its controls. As the theorem below shows, that choice is bookkeeping rather than hypothesis: it fixes which operators the family contains, but the stationarity it establishes constrains perturbations that need not respect the resolution at all.

Recall from Lemma~\ref{lem:dagger} that we may remove an overall phase. Thus, we define the operator
\begin{equation}
       V(\bm{\delta}) = e^{-i\theta_{n}}U = P_{n}+\sum_{a=1}^{n-1}e^{i\delta_{a}}P_{a},\label{eq:specV}
\end{equation}
where $\delta_{a} = \theta_{a} - \theta_{n}$ are the relative eigenphases, so that $\delta_{n}=0$ and the independent ones are $\delta_{1},\ldots,\delta_{n-1}$, and so $\mathcal{E}_{p}(U) = \mathcal{E}_{p}(V)$. In the following, we introduce the shorthand notation $\mathcal{E}_{p}(\bm{\delta})\equiv\mathcal{E}_{p}(V(\bm{\delta}))$. With this setup, we are ready to prove the following result:
\begin{theorem}\label{thm:main}
For an arbitrary unitary $U$ on $\mathcal{H}=\HA\otimes\HB$ with $V(\bm{\delta})$ defined in Eq.~(\ref{eq:specV}):
\textit{(i)} $\mathcal{E}_{p}(\bm{\delta})=\mathcal{E}_{p}(-\bm{\delta})$ for all
$\bm{\delta}$; and \textit{(ii)} $\nabla_{\bm{\delta}}\,\mathcal{E}_{p}(\bm{\delta}_{\bm{m}})=0$
at every corner $\bm{\delta}_{\bm{m}}$, whose components are all 0 or $\pi$ modulo
$2\pi$. There are $2^{n-1}$ such corners.
\end{theorem}
\begin{proof}
For {(i)}, note that
$V(\bm{\delta})^{\dagger}=P_{n}+\sum_{a=1}^{n-1}e^{-i\delta_{a}}P_{a}=V(-\bm{\delta})$, so
that Lemma~\ref{lem:dagger} gives
$\mathcal{E}_{p}(\bm{\delta})=\mathcal{E}_{p}(-\bm{\delta})$ immediately.

For {(ii)}, since the eigenphases are defined modulo $2\pi$, the corners are precisely
the fixed points of inversion on the torus,
$-(\bm{\delta}_{\bm{m}})_{a}\equiv(\bm{\delta}_{\bm{m}})_{a}\ (\mathrm{mod}\ 2\pi)$. Combined
with {(i)}, a displacement $\bm{\epsilon}$ about $\bm{\delta}_{\bm{m}}$ then obeys
$\mathcal{E}_{p}(\bm{\delta}_{\bm{m}}+\bm{\epsilon})
=\mathcal{E}_{p}(\bm{\delta}_{\bm{m}}-\bm{\epsilon})$. The linear term of the Taylor
expansion of $\mathcal{E}_{p}$ about $\bm{\delta}_{\bm{m}}$ is odd under
$\bm{\epsilon}\to-\bm{\epsilon}$ and must therefore vanish. Counting is immediate: each of
the $n-1$ independent components takes one of two values, giving $2^{n-1}$ corners.
\end{proof}

We introduce $\bm{m}=(m_1,\ldots, m_{n-1})$ with $m_{a}\in\{0,1\}$, taken modulo 2, to label the $2^{n-1}$ corners: $\bm{\delta}_{\bm{m}} = \pi \bm{m}$. At each corner $\bm{m}$ one can define a unitary operator -- the corner gate --  as
\begin{equation}
  R_{\bm{m}}\equiv V(\bm{\delta}_{\bm{m}})=P_{n}+\sum_{a=1}^{n-1}(-1)^{m_{a}}P_{a}=\mathbb{I}-2Q_{\bm{m}},
  \label{eq:reflection}
\end{equation}
where $Q_{\bm{m}}=\sum_{a:\,m_{a}=1}P_{a}$ is a projector. Hence every corner gate is a Hermitian involution, or generalized reflection: $R_{\bm{m}}^{\dagger}=R_{\bm{m}}$ and $R_{\bm{m}}^{2}=\mathbb{I}$. Since the $P_{a}$ are orthogonal and the components of $\bm{m}$ are taken modulo 2, one sees $R_{\bm{m}}R_{\bm{m}'}=R_{\bm{m}+\bm{m}'}$ and the corner gates form the finite group $\mathbb{Z}_{2}^{\,n-1}$, which is the direct product of $(n-1)$ copies of $\mathbb{Z}_2=\{0,1\}$. Restoring the overall phase, $U_{\bm{m}}\equiv e^{i\theta_{n}}R_{\bm{m}}$ satisfies $U_{\bm{m}}^{2}=e^{2i\theta_{n}}\mathbb{I}$: the physical gate is an involution in projective unitary space. The reflection $R_{\bm{m}}$ has eigenvalues $\pm 1$ with eigen-projectors $\Pi_{+}=\mathbb{I}-Q_{\bm{m}}$ and $\Pi_{-}=Q_{\bm{m}}$, since $R_{\bm{m}}=(\mathbb{I}-Q_{\bm{m}})-Q_{\bm{m}}$. 
%We call Theorem~\ref{thm:main} together with Eq.~\eqref{eq:reflection} the universal stationary-reflection theorem.

%Two cautions. The count $2^{n-1}$ bounds the corners, not the inequivalent gates; for two qubits all eight corners have vanishing entangling power. And a reflection need not be entangling. What matters is how the reflected subspace sits relative to the factorization, not its rank: the rank-one \textit{product} projector $\ketbra{11}{11}$ gives the maximally entangling controlled-$Z$, while the rank-one \textit{maximally entangled} projector gives the non-entangling SWAP.

Theorem~\ref{thm:main} locates the guaranteed stationary points but does not characterize them as minima, maxima or saddle points. Nor does it exhaust the stationary set: additional extrema generically occur outside of the corners. One can, however, compute the value of $\mathcal{E}_{p}$ at every corner, which we derive next. %Finally, since the entangling power is invariant under Hermitian conjugation and overall phase of the operator $U(\theta_{1},\theta_{2},\ldots,\theta_{n})$, the theorem also holds under the transformation $\theta_{a}\to\varphi-\theta_{a}$ for any fixed $\varphi$, in which case the corners are located at $\theta_a\in\{\varphi/2,\,\varphi/2+\pi\}$.

%Finally, let us remark that the choice of tensor product decomposition, $\mathcal{H}=\HA\otimes\HB$, underlying the calculation of the entangling power is arbitrary. This subtlety of Theorem~\ref{thm:main} is highly nontrivial as even the definition of entangled states actually depends on the choice of tensor product decomposition~\cite{Zanardi:2004zz}, and a given state in $\mathcal{H}$ can be entangled or separable depending on this choice. Nevertheless, Theorem~\ref{thm:main} holds for \textit{any} choice of tensor product decomposition.

%%%%%%%%%%%%%%%%%%%%%%%%%%%%%%%%%%%%%%%%%%%%%%%%%%%%%%%%%%%%%%%%%%%%%%%%
%\subsection{\label{sec:cornervalue}The entangling power at a corner}
%%%%%%%%%%%%%%%%%%%%%%%%%%%%%%%%%%%%%%%%%%%%%%%%%%%%%%%%%%%%%%%%%%%%%%%%

Every corner is a reflection, so its entangling power depends on the single projector $Q$ of
Eq.~\eqref{eq:reflection}. Although $\mathcal{E}_{p}$ is invariant under the  two-sided
local unitary group, $V\to WVW'$ with $W=u_{A}\otimes u_{B}$ and $W'=u'_{A}\otimes u'_{B}$,  on $V=\mathbb{I}-2Q$ we have:
\begin{equation}
  W\,(\mathbb{I}-2Q)\,W' = \big(\mathbb{I}-2\,WQW^{\dagger}\big)\,WW' .
  \label{eq:twosided}
\end{equation}
Since $WW'$ is again a local unitary which does not modify the entangling power, $\mathcal{E}_{p}(\mathbb{I}-2Q)$ must be  a function of the invariants of $Q$ under
$Q\to WQW^{\dagger}$. Since the realignment and the partial transpose are {linear} operations,
$I_{R}$ and $I_{T}$ are quartic in the operator and so of degree at most four in $Q$; the entangling power is then 
 is a polynomial of degree four in those invariants.

Consider first the invariants generated by $I_R$. Recall that the definition in Eq.~(\ref{eq:Rdef}) that $U^R$ maps $\HB\otimes\HB$ to $\HA\otimes\HA$.
%its row index is a pair of $A$ labels and its column index a pair of $B$ labels, 
So $V^{R}$ is a
$d_{A}^{2}\times d_{B}^{2}$ matrix rather than an operator on $\mathcal{H}$. This motivates introducing
% Identifying $\HA\otimes\HA$ with the operators on $\HA$, a vector with components $x_{(ij)}$ is the matrix $x$ and inner products are Hilbert--Schmidt products. In that language
$\ket{\delta_{A}}=\sum_{i}\ket{ii}$ and $\ket{\delta_{B}}=\sum_{\alpha}\ket{\alpha\alpha}$ in $\HA\otimes\HA$ and $\HB\otimes\HB$, respectively, such that %are the identity operators $\mathbb{I}_{A}$ and $\mathbb{I}_{B}$, with
$\braket{\delta_{A}|\delta_{A}}=d_{A}$ and $\braket{\delta_{B}|\delta_{B}}=d_{B}$, and
$\mathbb{I}^{R}=\ket{\delta_{A}}\bra{\delta_{B}}$.
Then $V^{R}=\ket{\delta_{A}}\bra{\delta_{B}}-2Q^{R}$. $I_R$ then depends on the following traces
\bea
 \!\!\!\!\!\!\! \bra{\delta_{A}}Q^{R}\ket{\delta_{B}}&=&r, 
\qquad \bra{\delta_{B}}Q^{R\dagger}Q^{R}\ket{\delta_{B}}=\mu_{A},\notag\\
 \!\!\!\! \!\! \! \bra{\delta_{A}}Q^{R}Q^{R\dagger}\ket{\delta_{A}}&=&\mu_{B},
 \quad\bra{\delta_{A}}Q^{R}Q^{R\dagger}Q^{R}\ket{\delta_{B}}=\kappa,
\eea
%with $v_{(ij)}=(\rho_{A})_{ij}$ and $w_{(\alpha\beta)}=(\rho_{B})_{\beta\alpha}$, so $\braket{\delta_{A}|v}=r$, $\braket{v|v}=\mu_{A}$ and $\braket{w|w}=\mu_{B}$. The third identity is what produces $\kappa$. 
which can be expressed as the invariants of $Q$ as 
\begin{align}
  &r=\Tr Q,\qquad \mu_{A}=\Tr\rho_{A}^{2},\qquad \mu_{B}=\Tr\rho_{B}^{2},\notag\\
  &\kappa=\Tr\big[Q\,(\rho_{A}\otimes\rho_{B})\big],
  \label{eq:invA}
\end{align}
where we have defined the partial traces $\rho_{A}=\Tr_{B}Q$ and $\rho_{B}=\Tr_{A}Q$.
Expanding $\Tr[(V^{R}V^{R\dagger})^{2}]$ gives
\begin{flalign}\nonumber
    I_R(\mathbb{I}-2Q) = d_{A}^{2}&d_{B}^{2}-8d_{A}d_{B}r + 8r^2 + 8d_{A}\mu_{A}\\
    &+8d_{B}\mu_{B} + 16I_{R}(Q)-32\kappa
\end{flalign}

For $I_T$, it is easy to see that the final expression depends simply  on the trace powers $q_{n}=\Tr\big[(Q^{T_{A}})^{n}\big]$, for $n=1-4$. Expanding $\Tr[(V^{T_{A}}V^{T_{A}\dagger})^{2}]$ we obtain
\begin{flalign}
    I_{T}(\mathbb{I}-2Q) = d_{A}d_{B} + 16(r + q_{4} -2q_{3})
\end{flalign}

The four of Eq.~\eqref{eq:invA} are all of the form $\Tr[Q\,X]$, with $X$ running over $\{\mathbb{I}_{A},\rho_{A}\}\otimes\{\mathbb{I}_{B},\rho_{B}\}$ and giving
$r$, $\mu_{A}$, $\mu_{B}$ and $\kappa$, respectively. Their invariance follows from the fact that tracing out one
factor absorbs any unitary acting on it, so that $\rho_{A}\to u_{A}\rho_{A}u_{A}^{\dagger}$
and $\rho_{B}\to u_{B}\rho_{B}u_{B}^{\dagger}$. The partial-transpose moments follow  from
\be
[(u_{A}\otimes u_{B})Q(u_{A}\otimes u_{B})^{\dagger}]^{T_{A}}
=(u_{A}^{*}\otimes u_{B})Q^{T_{A}}(u_{A}^{*}\otimes u_{B})^{\dagger}\ ,
\ee
again a local
conjugation.  Collect the invariants into
\begin{align}
  \alpha &\equiv d_{A}d_{B}\,r+r^{2}-d_{A}\mu_{A}-d_{B}\mu_{B},\label{eq:alpha}\\
  \beta &\equiv r^{2}+r+I_{R}(Q)+I_{T}(Q)-2\kappa-2q_{3}.\label{eq:beta}
\end{align}
We then find, for any $n$, any projector family, and any subsytem dimensions $d_{A}$ and $d_{B}$,
\begin{equation}
  \mathcal{E}_{p}(\mathbb{I}-2Q)=\frac{8}{D}\big(\alpha-2\beta\big),
  \label{eq:cornervalue}
\end{equation}
where $D\equiv d_{A}d_{B}(d_{A}+1)(d_{B}+1)$. %This is the $x=2$ case of Eq.~\eqref{eq:twochan} below. A reflection therefore sees the reflected subspace only through its rank, its two marginal purities, its overlap with the product of the marginals, its operator-Schmidt purity, and two partial-transpose moments.

It is interesting to observe that the objects $\alpha,\beta$ carry more symmetry than the invariants in Eqs.~(\ref{eq:invA}) and the $q_n$.
Replacing $Q$ by the complementary projector $\mathbb{I}-Q$ sends $R_{\bm{m}}$ to $-R_{\bm{m}}$. This operation preserves the entangling power but changes all seven invariants: $\{r, \mu_A, \mu_B, \kappa, I_R(Q), q_3, q_4\}$. Nevertheless,  $\alpha$
and $\beta$ remain unchanged. 

%The reason is that $R_{\bm{m}}$ and $-R_{\bm{m}}$ are the same gate up to a phase: in the two-channel family of Sec.~\ref{sec:twochannel} the replacement sends $V(\delta)\to e^{i\delta}V(-\delta)$, so by invariance under an overall phase together with Theorem~\ref{thm:main}\textit{(i)} the entire landscape is unchanged, which forces the two coefficients to agree term by term. It is accordingly $\alpha$ and $\beta$, not the seven ingredients, that are the invariants attached to a corner.

%%%%%%%%%%%%%%%%%%%%%%%%%%%%%%%%%%%%%%%%%%%%%%%%%%%%%%%%%%%%%%%%%%%%%%%%
\subsection{\label{sec:twochannel}Exact solution for two-channel dynamics}
%%%%%%%%%%%%%%%%%%%%%%%%%%%%%%%%%%%%%%%%%%%%%%%%%%%%%%%%%%%%%%%%%%%%%%%%
For $n=2$ the entangling power only depends on one relative phase $\delta$ and $\mathcal{E}_{p}(\delta)$ can be computed exactly on the entire circle, $\delta\in [0,2\pi)$, not just at the corner $\delta_{\bm{m}} = \{0,\pi\}$.
 Let $Q$ be any orthogonal projector and
\begin{equation}
  V(\delta)=e^{i\delta}Q+(\mathbb{I}-Q)=\mathbb{I}+c\,Q,\qquad c=e^{i\delta}-1 .
  \label{eq:twochanV}
\end{equation}
Writing $x\equiv1-\cos\delta\in[0,2]$, %so that $c+\bar{c}=-2x$, $|c|^{2}=2x$, $c^{2}+\bar{c}^{2}=4x^{2}-4x$ and $|c|^{4}=4x^{2}$, 
we obtain
\begin{equation}
  \mathcal{E}_{p}(\delta)=\frac{4}{D}\,x\,\big(\alpha-\beta x\big).
  \label{eq:twochan}
\end{equation}
The entangling power of \textit{any} two-projector unitary is a downward parabola in
$1-\cos\delta$, degenerating to a straight line when $\beta=0$.

The derivation is short. For the partial transpose, $V^{T_{A}}=\mathbb{I}+cQ^{T_{A}}$ with $Q^{T_{A}}$ Hermitian, so $V^{T_{A}}V^{T_{A}\dagger}=\mathbb{I}+(c+\bar{c})Q^{T_{A}}+|c|^{2}(Q^{T_{A}})^{2}$. %and
%
%\begin{align}
%  I_{T}=\;&d_{A}d_{B}+\big[(c+\bar{c})^{2}+2|c|^{2}\big]\Tr (Q^{T_{A}})^{2}\notag\\
 % &+2(c+\bar{c})\Tr Q^{T_{A}}+2(c+\bar{c})|c|^{2}q_{3}\notag\\
%  &+|c|^{4}q_{4}.
%\end{align}
%
Using $\Tr Q^{T_{A}}=\Tr (Q^{T_{A}})^{2}=r$ as above, we arrive at
\begin{equation}
  I_{T}=d_{A}d_{B}+4x^{2}\big(r+q_{4}-2q_{3}\big).
  \label{eq:ITtwochan}
\end{equation}
%
%This cancellation is general; Sec.~\ref{sec:twoproj} exhibits it in a special case. 
For the realignment, we can use the trace identities above 
and expanding $\Tr[(V^{R}V^{R\dagger})^{2}]$ gives
\begin{equation}
  I_{R}=d_{A}^{2}d_{B}^{2}-4x\,\alpha+4x^{2}\big(r^{2}+I_{R}(Q)-2\kappa\big),
  \label{eq:IRtwochan}
\end{equation}
and substituting Eqs.~\eqref{eq:ITtwochan} and \eqref{eq:IRtwochan} into Eq.~\eqref{eq:master} yields Eq.~\eqref{eq:twochan}.

A few observations, all concerning $\mathcal{E}_{p}$ along the circle of Eq.~\eqref{eq:twochanV}:

(a) {\em Region near the Identity gate --} Near $\delta=0$, $V(\delta)\simeq \mathbb{I}$, $x\simeq\delta^{2}/2$, and $\mathcal{E}_{p}\simeq(2\alpha/D)\,\delta^{2}$. The single invariant $\alpha$ fixes the behavior near the Identity gate. Positivity forces $\alpha\ge0$.

(b) {\em Corner behavior -- } At $\delta=\pi$, $\mathcal{E}_{p}(\pi)=\tfrac{8}{D}(\alpha-2\beta)\ge0$ recovering Eq.~\ref{eq:cornervalue}. Note, this implies, $\alpha\ge2\beta$ always, and
\begin{equation}
  \mathcal{E}_{p}''(\pi)=-\frac{4}{D}\big(\alpha-4\beta\big).
\end{equation}
The entangling power at the corner is a maximum if $\alpha\ge 4\beta$,\footnote{At $\alpha= 4\beta$, $\mathcal{E}_{p}''(\pi)=0$ and the maximum is determined by the fourth-order Taylor expansion.} a minimum if $\alpha<4\beta$.  %At $n\ge3$ the  requires the full Hessian.

(c) {\em Away from the corners -- } If $\alpha<4\beta$ the maximum of $\mathcal{E}_{p}$ on the circle is away from a corner, in the interior at $x_{*}=\alpha/2\beta$, with value $\alpha^{2}/(\beta D)$; otherwise the corner is the maximum. 

(d) {\em Non-entangling corner -- } The conditions $\mathcal{E}_{p}(\pi)=0$, $\alpha=2\beta$ and $\mathcal{E}_{p}(\delta)=(2\alpha/D)\sin^{2}\delta$ are equivalent: the corner is non-entangling if and only if the entangling power is a pure $\sin^{2} \delta$.

%%%%%%%%%%%%%%%%%%%%%%%%%%%%%%%%%%%%%%%%%%%%%%%%%%%%%%%%%%%%%%%%%%%%%%%%
\subsection{\label{sec:gates}Classification of corner gates}
%%%
\begin{table*}[t]
\caption{\label{tab:gates}Standard quantum gates and the corner criterion, grouped by
whether they lie in the Clifford group. $Q$ is the projector of the reflection
$\mathbb{I}-2Q$, $r=\Tr Q$, and $\alpha,\beta$ are from Eqs.~\eqref{eq:alpha} and
\eqref{eq:beta}; see the text for how the entries are read. In the third block, where the
gates are not corners, the $Q$ column lists the distinct eigenvalues instead. Here
$\ket{\pm}=(\ket{0}\pm\ket{1})/\sqrt{2}$ and $P_{s}$ is the two-qubit singlet projector.}
\begin{ruledtabular}
\footnotesize
\begin{tabular}{lccccccc}
gate & $d_A\!\times\!d_B$ & $Q$ & $r$ & $\alpha$ & $\beta$ & $\mathcal{E}_p$ & type \\
\colrule
\multicolumn{8}{l}{\textit{Clifford corner gates}}\\
controlled-$Z$ & $2\times2$ & $\ketbra{11}{11}$ & $1$ & $1$ & $0$ & $2/9$ & max\footnotemark[1] \\
${\rm CNOT}$ & $2\times2$ & $\ketbra{1}{1}\otimes\ketbra{-}{-}$ & $1$ & $1$ & $0$ & $2/9$ & max\footnotemark[1] \\
Grover diffusion & $2\times2$ & $\ketbra{++}{++}$ & $1$ & $1$ & $0$ & $2/9$ & max\footnotemark[1] \\
${\rm SWAP}$ & $2\times2$ & $P_{s}$ & $1$ & $3$ & $3/2$ & $0$ & min \\
$\sigma_{k}\otimes\sigma_{k}$ & $2\times2$ & $\tfrac12(\mathbb{I}-\sigma_{k}\otimes\sigma_{k})$ & $2$ & $4$ & $2$ & $0$ & min\footnotemark[2] \\
\colrule
\multicolumn{8}{l}{\textit{non-Clifford corner gates}}\\
controlled-controlled-$Z$ & $2\times4$ & $\ketbra{111}{111}$ & $1$ & $3$ & $0$ & $1/5$ & max\footnotemark[3] \\
Toffoli & $2\times4$ & $\ketbra{11}{11}\otimes\ketbra{-}{-}$ & $1$ & $3$ & $0$ & $1/5$ & max\footnotemark[3] \\
Fredkin & $2\times4$ & $\ketbra{1}{1}\otimes P_{s}$ & $1$ & $3$ & $0$ & $1/5$ & max\footnotemark[3] \\
Fredkin & $4\times2$ & $\ketbra{1}{1}\otimes P_{s}$ & $1$ & $6$ & $3/2$ & $1/5$ & marginal\footnotemark[4] \\
$C^{n-1}Z$ & $2\times2^{n-1}$ & $\ketbra{1\cdots1}{1\cdots1}$ & $1$ & $2^{n-1}\!-\!1$ & $0$ &
 $\dfrac{4(2^{n-1}-1)}{3\cdot2^{n-1}(2^{n-1}+1)}$ & max\footnotemark[5] \\
\colrule
\multicolumn{8}{l}{\textit{not corner gates}}\\
$S$ & --- & $\{1,i\}$ & --- & --- & --- & $0$ & ---\footnotemark[6] \\
$i{\rm SWAP}$ & $2\times2$ & $\{1,\pm i\}$ & --- & --- & --- & $2/9$ & --- \\
$T$ & --- & $\{1,e^{i\pi/4}\}$ & --- & --- & --- & $0$ & ---\footnotemark[6] \\
controlled-$S$ & $2\times2$ & $\{1,i\}$ & --- & --- & --- & $1/9$ & --- \\
controlled-$T$ & $2\times2$ & $\{1,e^{i\pi/4}\}$ & --- & --- & --- & $(2-\sqrt2)/18$ & --- \\
$\sqrt{{\rm SWAP}}$ & $2\times2$ & $\{1,i\}$ & --- & --- & --- & $1/6$ & --- \\
\end{tabular}
\footnotetext[1]{Also the global maximum over all two-qubit gates.}
\footnotetext[2]{A product of single-qubit unitaries, hence non-entangling.}
\footnotetext[3]{A maximum along its own circle but a saddle point on the full unitary group, with one descending direction.}
\footnotetext[4]{$\alpha=4\beta$: the quadratic term vanishes and the type is fixed at fourth order, where the corner is a quartic maximum.}
\footnotetext[5]{Clifford only for $n=2$; the entry gives $2/9$, $1/5$, $7/54$ and $5/68$ for $n=2,3,4,5$.}
\footnotetext[6]{A single-qubit gate, hence local, with vanishing entangling power for any bipartition.}
\end{ruledtabular}
\end{table*}
%%%

Whether a given gate can occur as a ``corner gate'' has a complete and elementary answer: a particular quantum gate
$U$ can be realized as a
corner gate of some projector family if and only if $U^{2}\propto\mathbb{I}$. Necessity is immediate: a corner gate is $e^{i\theta_{n}}R_{\bm{m}}$ with $R_{\bm{m}}^{2}=\mathbb{I}$, so $U^{2}=e^{2i\theta_{n}}\mathbb{I}$. For sufficiency, note first that the proportionality constant is unimodular because $U^2$ is itself a unitary matrix: $U^2=c \,\mathbb{I}$ and $(c\mathbb{I})^\dagger(c\mathbb{I})=\mathbb{I}$ implies $|c|^2=1$. Then one can write $U^{2}=e^{2i\varphi}\mathbb{I}$ and
$W\equiv e^{-i\varphi}U$ is unitary and satisfies $W^{2}=\mathbb{I}$, hence is Hermitian, $W^{\dagger}=W^{\dagger}W^{2}=W$, and is therefore itself a generalized reflection. A general reflection has eigenvalues drawn from  $\pm 1$ and we denote the  projectors onto the $-1$ and $+1$ eigenspaces by $Q$ and $\mathbb{I}-Q$, respectively, which leads to $W=\mathbb{I}-2Q$.  Therefore  $W$ is a corner of every family whose projectors can be summed to give  $Q$. That is, each projector lies entirely inside $Q{\cal H}$ or $(1-Q){\cal H}$. The same holds for $U=e^{i\varphi}W$, with $\varphi$ in the role of the reference phase $\theta_{n}$ in Eq.~(\ref{eq:reflection}).

%Therefore $W$ is  acorner of any resolution refining $\{Q,\mathbb{I}-Q\}$. Equivalently, $U$ has at most twodistinct eigenvalues and their ratio is $-1$. This is weaker than being a corner of a \textit{given} family, which additionally requires $Q$ to be a sum of that family's projectors and is what the count $2^{n-1}$ enumerates.

%The criterion cuts across the Clifford hierarchy rather than along it. 
The corner gates
include the Clifford entanglers: the controlled-$Z$, ${\rm CNOT}=\mathbb{I}-2\,
\ketbra{1}{1}\otimes\ketbra{-}{-}$, the SWAP, and products of Pauli operators, but also
gates outside the Clifford group: the Toffoli gate is $\mathbb{I}-2\,\ketbra{11}{11}\otimes
\ketbra{-}{-}$, the controlled-controlled-$Z$ is $\mathbb{I}-2\,\ketbra{111}{111}$, the
Fredkin gate is $\mathbb{I}-2\,\ketbra{1}{1}\otimes P_{s}$ with $P_{s}$ the two-qubit
singlet projector, and every multiply controlled $Z$, as well as every Boolean phase
oracle, is of this form. Excluded are the gates whose eigenvalue ratios are finer roots of
unity: the $T$ gate and its controlled versions, the phase gate $S$, $\sqrt{\rm SWAP}$, and 
$i{\rm SWAP}$. 

Table~\ref{tab:gates} applies the criterion to the standard gates, grouped by whether they
lie in the Clifford group.
Each corner gate is treated there as its own two-channel family, so $n=2$ throughout and the
type of the corner follows from the criterion $\alpha\gtrless4\beta$ of
Sec.~\ref{sec:twochannel}; gates sharing $(\alpha,\beta)$ within a block are locally
equivalent. The gates in the third block are not corner gates.
% the reflection data do not apply, since their eigenvalue ratios are not $-1$, but the entangling power remains defined and is given.

The distinction between the Toffoli and $T$ gates is worth stating explicitly, since both are
standard non-Clifford resources: the Toffoli gate is a corner because it is self-inverse,
whereas $T$ is not, with $T^{2}=S$ and $T^{8}=\mathbb{I}$. %Being a single-qubit gate, $T$ has vanishing entangling power in any case; 
Other non-Clifford entanglers excluded by the
criterion are the controlled-$S$ and controlled-$T$ gates. In fact the Toffoli and Hadamard gates alone are universal for quantum
computation~\cite{Shi:2002cli,Aharonov:2003hjn}, and the Hadamard gate is itself a Hermitian
involution, hence also a corner gate. There is therefore a universal gate set consisting
entirely of corner gates, even though the canonical magic gate $T$ is not one. 
\section{\label{sec:apps} Applications}

We now apply Theorem~\ref{thm:main} to a variety of physical systems. The examples
cover spectral resolutions with $n=2,3,$ and $4$ projectors, and their corners include
minima, maxima, and saddle points. Moreover, Eq.~\eqref{eq:cornervalue} gives the
entangling power at each corner without further computation.

%%%

\subsection{Two-qubit gates}
As a first application let us consider the simplest quantum mechanical example of two-qubit gates: $d_A=d_B=2$. In particular we consider the family of unitaries whose eigenvectors  are the Bell states
$\{\ket{\Phi^\pm},\ket{\Psi^\pm}\}$, with
$\ket{\Phi^{\pm}}=(\ket{00}\pm\ket{11})/\sqrt2$ and
$\ket{\Psi^{\pm}}=(\ket{01}\pm\ket{10})/\sqrt2$,
\begin{equation}
  U_B(\lambda)=\sum_{a} e^{i\lambda_a}B_a,\qquad B_a=\ket{a}\!\bra{a},
  \label{eq:belldiag}
\end{equation}
where  $a\in\{\Phi^+,\Phi^-,\Psi^+,\Psi^-\}$. In other words, $U_B$ is diagonal in the Bell basis and its entangling  power  is  given by the three relative phases between the Bell states, e.g. $\delta_a=\lambda_a-\lambda_{\Psi^-}$, on a
three-torus. Theorem~\ref{thm:main} then guarantees $2^{3}=8$ stationary corners.

Note that $U_B$ captures the entanglement dynamics of all two-qubit logic gates. The most general two-qubit gate is represented by $U\in SU(4)$ but  entanglement and entangling power are invariant under local unitaries, $U\to (u_1\otimes u_2)\,U\,(u_3\otimes u_4)$. So the space of inequivalent entanglers is characterized by the coset $SU(4)/SU(2)\otimes SU(2)$ which, using the Cartan decomposition, is parameterized by 
\be
\label{eq:ud}
U_d=e^{i\beta_k \sigma_k\otimes\sigma_k}\ , 
\ee
also on a 3-torus \cite{Zhang:2003zz,Low:2021ufv}. The eigenbasis of $U_d$ is exactly the Bell basis, with eigenvalues $e^{i\lambda_a}$ given by
\bea
\label{eq:belleig}
\lambda_ {\Phi^+}&=& \beta_1- \beta_2 + \beta_3 \ ,\quad   \lambda_{\Phi^-} =-\beta_1 + \beta_2 + \beta_3 \nonumber\\
\lambda_{\Psi^+} &=& \beta_1 + \beta_2 - \beta_3 \ , \quad 
\lambda_{\Psi^-} =-\beta_1-\beta_2- \beta_3\ .
\eea
The entanglement dynamics of all inequivalent two-qubit gates is therefore captured by Eq.~\eqref{eq:belldiag}, which has another form
\begin{equation}
U_B=\sum_{k=0}^{3}g_k\,\sigma_k\otimes\sigma_k
\end{equation}
with $\sigma_0=I$. In the above, $g_k$ can be written as
\begin{equation}
\begin{pmatrix} g_0 \\ g_1\\g_2\\g_3 \end{pmatrix}
 = \frac14
  \begin{pmatrix}
    \phantom{-}1 & \phantom{-}1 & \phantom{-}1 & \phantom{-}1\\
    \phantom{-}1 & -1 & \phantom{-}1 & -1\\
    -1 & \phantom{-}1 & \phantom{-}1 & -1\\
    \phantom{-}1 & \phantom{-}1 & -1 & -1
  \end{pmatrix} 
\begin{pmatrix} 
e^{i\lambda_{\Phi^+}} \\ e^{i\lambda_{\Phi^-}} \\ e^{i\lambda_{\Psi^+}} \\ e^{i\lambda_{\Psi^-}} \end{pmatrix}
  \ .
  \label{eq:signtable}
\end{equation}
Using this form to evaluate the invariants of Eq.~\eqref{eq:master} are particularly simple. We find 
\begin{equation}
  I_R=16\sum_{k=0}^{3}|g_k|^4\ , \qquad  I_T=\sum_a|\chi_a|^4\ ,
  \label{eq:bellIR}
\end{equation}
where
\begin{equation}
\begin{pmatrix} \chi_{\Phi^+} \\ \chi_{\Phi^-} \\ \chi_{\Psi^+} \\ \chi_{\Psi^-} \end{pmatrix}
 =\frac12
  \begin{pmatrix}
    \phantom{-}1 & \phantom{-}1 & \phantom{-}1 & -1\\
    \phantom{-}1 & \phantom{-}1 & -1 & \phantom{-}1\\
    \phantom{-}1 & -1 & \phantom{-}1 & \phantom{-}1\\
    -1 & \phantom{-}1 & \phantom{-}1 & \phantom{-}1
  \end{pmatrix}
\begin{pmatrix} e^{i\lambda_{\Phi^+}} \\ e^{i\lambda_{\Phi^-}} \\ e^{i\lambda_{\Psi^+}} \\ e^{i\lambda_{\Psi^-}} \end{pmatrix} \ .
  \label{eq:bellchi}
\end{equation}
The entangling power for $U_B$ then takes the closed form
\begin{equation}
  \mathcal{E}_p(\lambda)=\frac59-\frac{1}{36}\Big(16\sum_k|g_k|^4+\sum_a|\chi_a|^4\Big).
  \label{eq:bellEp}
\end{equation}

To investigate the nature of the corners, recall that up to an overall phase we may write
\begin{equation}
    R_{S}=\mathbb{I}-2Q_{S}
\end{equation}
with $Q_{S}=\sum_{a\in S}B_a$ for a subset $S\subseteq\{\Phi^+,\Phi^-,\Psi^+\}$. The eight possible subsets thus correspond to the 8 corners of the torus of Theorem~\ref{thm:main}. For $S=\varnothing$, $R_{S}=\mathbb{I}$; for $|S|=2$ one finds
$R_{S}=\pm\,\sigma_k\otimes\sigma_k$, which is a local unitary in the equivalence class of $\mathbb{I}$; and for $|S|=1$ or $3$,
$R_{S}=\pm(\mathbb{I}-2B_a)$ for a single Bell projector, which is locally equivalent to the SWAP. Thus, every one of
the eight corners is locally equivalent to $\mathbb{I}$ or SWAP, implying
\begin{equation}
  \mathcal{E}_p(\bm{\delta}_{\bm{m}})=0,\quad\forall\;\bm{m}.
\end{equation}
In terms of $\beta_k$ introduced in Eq.~(\ref{eq:ud}),   the relative eigenphases are
$\bm{\delta}=2(\beta_1+\beta_3,\;\beta_2+\beta_3,\;\beta_1+\beta_2)$ from
Eq.~\eqref{eq:belleig}, so a corner requires
$\beta_i+\beta_j=0$ modulo $\pi/2$ for every pair. Adding two of these conditions and
subtracting the third gives $2\beta_k\equiv0$ modulo $\pi/2$, so 
\begin{equation}
  \beta_k=n_k\frac{\pi}{4}\ ,\qquad n_i+n_j=0 \ \mathrm{mod}\ 2 .
  \label{eq:cornerbeta}
\end{equation}
Among the eight corners, four are locally equivalent to $\mathbb{I}$ and the other four are in the equivalence class of SWAP.  That is, all the stationary corner points for two-qubit gates are global minima  corresponding to  the Identity gate or the SWAP gate, which are the only two  two-qubit gates with vanishing entangling power \cite{Low:2021ufv}.\footnote{Interestingly, the corner conditions in Eq.~(\ref{eq:cornerbeta})  coincide with some of the vertices of Weyl chamber of $SU(4)/SU(2)\otimes SU(2)$, which was used to   classify the equivalence classes, under local unitaries, of  all two-qubit gates  in Ref.~\cite{Zhang:2003zz}.}

A physical system realizing $U_B$ in  two-qubit case is the two-site spin-1/2 spin-chain model with the Hamiltonian:
\be
\label{eq:2siteH}
H =a_x S_x\otimes S_x+a_y S_y\otimes S_y+a_z S_z\otimes S_z \ ,
\ee
where $S_k =\sigma^k/2$, $k=1,2,3$. Observe that its time-evolution operator $U(t)=e^{iHt}$ is exactly of the form $U_d$ introduced in Eq.~(\ref{eq:ud}). Therefore the eigenvalues of $H$ are  given by Eq.~(\ref{eq:belleig}), with the replacement $\beta_k\to a_k/4$, while the eigenvectors are maximally entangled Bell states, independent of $a_k$. Both the instantaneous and time-averaged entangling power of the time-evolution operator $U(t)=e^{iHt}$ were computed in Ref.~\cite{Low:2026oyf}. Here we recast the results there in light of Theorem ~\ref{thm:main}.

The time-evolution generated by Eq.~\eqref{eq:2siteH} follows along a straight line on the
torus: ${\beta}_k(t)=a_k t/4$. The instantaneous entangling power was obtained in closed
form in Ref.~\cite{Low:2026oyf}; which in the present variables reads
\begin{equation}
  \mathcal{E}_p(t)=\frac{2}{9}\Big[\,1-\prod_{k}\cos^{2}2\beta_k-\prod_{k}\sin^{2}2\beta_k\,\Big] .
  \label{eq:xyzEp}
\end{equation}
The bounds $0\le\mathcal{E}_p\le2/9$ are manifest, and both are attained on the same
quarter-turn lattice $\beta_k=n_k\pi/4$ that appears in Eq.~\eqref{eq:cornerbeta}, sorted by
parity: the entangling power vanishes when all three $n_k$ share an even/odd-parity, which is
Eq.~\eqref{eq:cornerbeta} itself, and reaches $2/9$ when two of the $\beta_k$ sit on the
lattice with opposite even/odd parities, the third being unconstrained.

The eight corners of Eq.~\eqref{eq:belldiag} are not the only reflections available on two
qubits, and the others need not be non-entangling. As an illustration, Ref.~\cite{Low:2026oyf} isolates the dependence of the entangling power on eigenvectors
alone by means of the isospectral Hamiltonians $H(\ket{\psi})=J(\mathbb{I}/4-\ketbra{\psi}{\psi})$,
which share the spectrum of the two-site XXX model for an arbitrary  normalized  $\ket{\psi}$. Their time
evolution, $e^{iHt}\propto\mathbb{I}+(e^{-iJt}-1)\ketbra{\psi}{\psi}$, is exactly the two-channel
family of Eq.~\eqref{eq:twochanV} with $Q=\ketbra{\psi}{\psi}$ and $\delta=-Jt$. It turns out the two invariants, $\alpha$ and $\beta$ in Eq.~\eqref{eq:twochan}, only depend on the concurrence $C$ \cite{Hill:1997pfa,Wootters:1997id} of $|\psi\rangle$,
\begin{equation}
  \alpha=1+2C^{2} ,\qquad \beta=C^{2}\Big(1+\frac{C^{2}}{2}\Big) .
  \label{eq:isoab}
\end{equation}
The instantaneous entangling power then follows from Eq.~(\ref{eq:twochan}). Moreover, the corner at $-Jt=\pi$ is the reflection
$\mathbb{I}-2\ketbra{\psi}{\psi}$, where Eq.~\eqref{eq:cornervalue} gives
\begin{equation}
  \mathcal{E}_{p}\big(\mathbb{I}-2\ketbra{\psi}{\psi}\big)=\frac{2}{9}\big(1-C^{4}\big) .
  \label{eq:rankonepsi}
\end{equation}
 \begin{figure*}[t]
    \centering
    \includegraphics[width=0.75\linewidth]{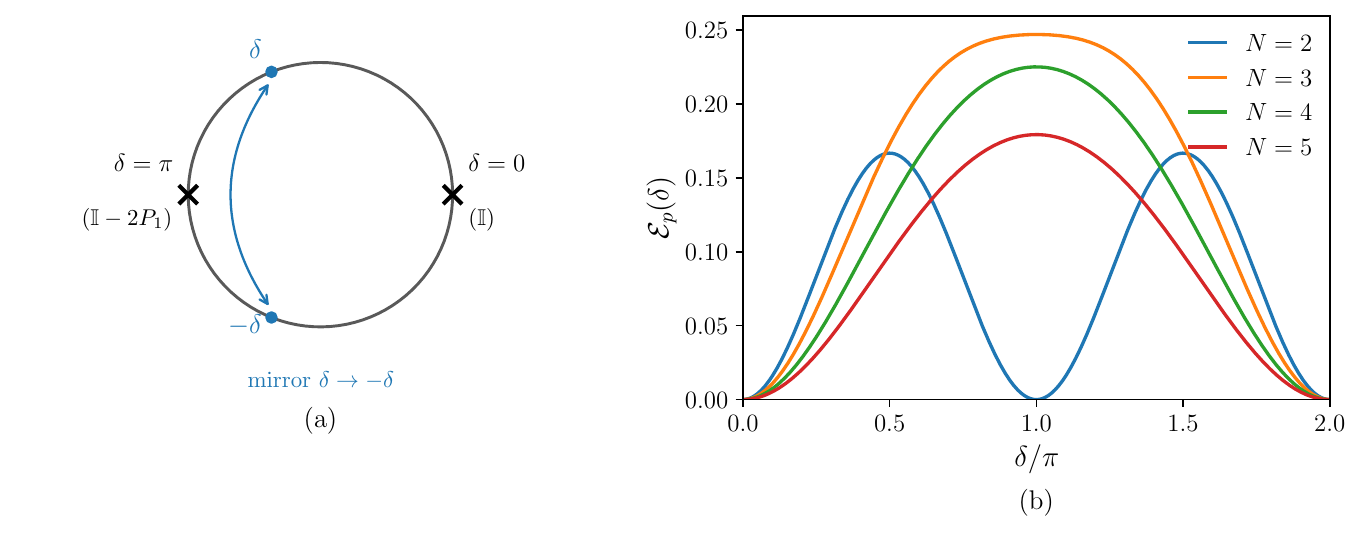}
    \caption{\label{fig:two_proj_fig}
    (a) The two-projector unitary circle, corresponding to the boundary
    of Fig.~1 in Ref.~\cite{McGinnis:2025iab} and parametrized by the
    relative phase $\delta$. The fixed points $\delta=0,\pi$ under
    $\delta\to-\delta$ are stationary by Theorem~\ref{thm:main}.
    (b) Entangling power of the
    singlet--adjoint family for $N=2,3,4,5$. The reflection at
    $\delta=\pi$ is a minimum for $N=2$ and a maximum for $N\ge3$.}
\end{figure*}
Although the entangling power at the corner is stationary, it need not vanish: if $|\psi\rangle$ is a product state, $C=0$,  $\mathcal{E}_{p}$ reaches the maximum at $2/9$, while if $|\psi\rangle$ is a Bell state, $C=1$, then $\mathcal{E}_{p}=0$.  The corner gate is locally equivalent to  the controlled-$Z$ gate  at $C=0$ and the SWAP at  $C=1$.

\subsection{Two-projector for $d_{A}=d_{B}=N$}
\label{sec:twoproj}

We now set $d_A=d_B=N$ and consider two projectors satisfying
$P_1+P_2=\mathbb{I}$. After removing an overall phase, the unitary
depends only on the relative phase $\delta$,
\begin{equation}
  V(\delta)=e^{i\delta}P_1+P_2,
  \label{eq:refl}
\end{equation}
with $V(0)=\mathbb{I}$ and $V(\pi)=P_2-P_1=\mathbb{I}-2P_1$.
Thus, the phase torus reduces to the circle in
Fig.~\ref{fig:two_proj_fig}(a), with stationary corners at
$\delta=0,\pi$ by Theorem~\ref{thm:main}.

The two decompositions below are the $SU(N)$-invariant channels for
fundamental--fundamental and fundamental--antifundamental particle scattering.
Ref.~\cite{McGinnis:2025brt,McGinnis:2025iab} studied their projector decompositions, while Ref.~\cite{McGinnis:2025xgt} analyzed crossing
and separability. Here we apply Theorem~\ref{thm:main} to the
 entangling power on the  circle $\delta\in[0, 2\pi)$. The fundamental--fundamental channel decomposes into the symmetric and antisymmetric irreducible representations, $\mathbf{N}\otimes\mathbf{N}=\mathbf{S}\oplus\mathbf{A}$, whose two corners are global minima. The fundamental--antifundamental channel is the
singlet--adjoint case, $\mathbf{N}\otimes\overline{\mathbf{N}}=\mathbf{1}\oplus\mathbf{adj}$, and the $\delta=\pi$ corner is instead a maximum for $N\ge3$.

\subsubsection{$\mathbf{N}\otimes\mathbf{N}=\mathbf{S}\oplus\mathbf{A}$}
\label{sect:nnsa}

The channel projectors are
\bea
  P_1=P_S &\equiv& \frac12\left(\mathbb{I}+S_W\right)\ ,\nonumber\\
  P_2=P_A &\equiv& \frac12\left(\mathbb{I}-S_W\right)\ ,
  \label{eq:SAproj}
\eea
where $S_W$ is the SWAP operator interchanging the two qudits.
The relevant invariants of Eqs.~\eqref{eq:alpha} and \eqref{eq:beta} are
$\alpha=N^2(N^2-1)/4$ and $\beta=\alpha/2$. Therefore,
\begin{equation}
  \mathcal{E}_p(\delta)
  =\frac{N-1}{2(N+1)}\,\sin^2\delta .
  \label{eq:Eexch}
\end{equation}
The corners $V(0)=\mathbb{I}$ and $V(\pi)=-S_W$ are, up to an overall
phase, the Identity gate and the SWAP gate, and hence global minima with
$\mathcal{E}_p=0$.

Substituting Eq.~\eqref{eq:SAproj} into Eq.~\eqref{eq:refl} gives
\bea
  V(\delta)&=&e^{i\delta/2}
  \left(\cos\frac{\delta}{2}\,\mathbb{I}
  +i\sin\frac{\delta}{2}\,S_W\right) \nonumber \\
  &=&e^{i\delta/2}e^{i(\delta/2)S_W},
  \label{eq:vsw}
\eea
where the last equality follows from $S_W^2=\mathbb{I}$. Up to the
overall phase $e^{i\delta/2}$, Eq.~\eqref{eq:vsw} is the qudit extension
of the partial-SWAP family introduced for qubits in
Refs.~\cite{Scarani:2002djz,Ziman:2002uuo}, with parameter
$\eta=\delta/2$. For $N=2$, Eq.~\eqref{eq:Eexch} agrees with the
$(\mathrm{SWAP})^\nu$ result of Ref.~\cite{Fan:2005opw} upon setting
$\nu=-\delta/\pi$. For arbitrary $N$, it agrees with
Ref.~\cite{Jonnadula2020}, with $t=\delta/2$, after undoing the
normalization of the entangling power used there.

\subsubsection{$\mathbf{N}\otimes\overline{\mathbf{N}}=\mathbf{1}\oplus\mathbf{adj}$}

Let $\{\ket{i}\}_{i=1}^{N}$ be an orthonormal basis of the fundamental
space and $\{\ket{\bar i}\}_{i=1}^{N}$ the corresponding conjugate
basis of the antifundamental space. Define
\begin{equation}
  \begin{aligned}
    \ket{\Phi}
      &\equiv \frac{1}{\sqrt N}\sum_{i=1}^{N}
        \ket{i}\otimes\ket{\bar i},\\
    P_\Phi
      &=\ketbra{\Phi}{\Phi},\qquad
        P_{\mathrm{adj}}=\mathbb{I}-P_\Phi .
  \end{aligned}
  \label{eq:singletproj}
\end{equation}
The reduced unitary is $V(\delta)=e^{i\delta}P_\Phi+P_{\mathrm{adj}}$.
For $Q=P_\Phi$,
one finds $\alpha=N^2-1$ and  $\beta=2-{2}/{N^2}$.
Eq.~\eqref{eq:twochan} then gives
\begin{equation}
  \mathcal{E}_p(\delta)
  =\frac{4(N-1)}{N^2(N+1)}\,x
   \left(1-\frac{2x}{N^2}\right)  \ ,
  \label{eq:Esa}
\end{equation}
where $x=1-\cos\delta$.
At the corner $\delta=\pi$,
\begin{equation}
 \begin{aligned}
  \mathcal{E}_p(\pi)
  &=\frac{8(N-1)(N^2-4)}{N^4(N+1)},\\
  \mathcal{E}_p''(\pi)
  &=\frac{4(N-1)(8-N^2)}{N^4(N+1)} .
 \end{aligned}
  \label{eq:sacorner}
\end{equation}
As a function of $x\in[0,2]$, Eq.~\eqref{eq:Esa} is a concave
quadratic with vertex $x_*=N^2/4$. For $N=2$, the corner at
$\delta=\pi$ is a global minimum. For every integer $N\ge3$, one has
$x_*>2$, so the corner is a global maximum, as illustrated in
Fig.~\ref{fig:two_proj_fig}(b).

For $N=2$, pseudoreality implies $\overline{\mathbf{2}}$ is unitarily equivalent to $\mathbf{2}$ and the situation reduces to that discussed in Sec.~\ref{sect:nnsa}.

\subsection{Three-projectors for $d_{A}=d_{B}=3$ }
\label{sec:threeproj}

The two-projector examples above reduce the relative-phase torus to a circle. We now turn to the simplest genuinely higher-dimensional case, with three projectors and two independent relative phases. A natural realization is the two-site spin-1 chain model, which generalizes the two-site spin-1/2 model in Eq.~(\ref{eq:2siteH}). Angular-momentum addition decomposes the total 
Hilbert space into sectors of total spin $j=0,1,2$:
\begin{equation}
  \mathbf{3}\otimes\mathbf{3}
  =\mathbf{1}\oplus\mathbf{3}\oplus\mathbf{5},
  \qquad
  P_j\equiv\sum_{m=-j}^{j}\ketbra{j,m}{j,m}.
  \label{eq:spin1proj}
\end{equation}
Here $P_j$ projects onto the total-spin-$j$ sector, $m$ is the
total-spin projection in the $\hat{z}$ direction, and $\sum_jP_j=\mathbb{I}$. The boldface
numbers in Eq.~\eqref{eq:spin1proj} denote the dimensions $2j+1$ of
the corresponding $SU(2)$ representations.

For simplicity we assume rotational invariance, which requires
the Hamiltonian to be a function of $\mathbf{J}^2$, where $\mathbf{J}$ is the total angular momentum. For the two-site spin-$1$ chain, the total-spin quantum number takes the values $j=0,1,2$, and $\mathbf{J}^2$ has the three distinct eigenvalues
$j(j+1)=0,2,6$. It therefore satisfies the minimal polynomial,\footnote{The full characteristic polynomial has its powers weighted by the eigenvalue degeneracies: $\mathbf{J}^2
  \left(\mathbf{J}^2-2\mathbb{I}\right)^3
  \left(\mathbf{J}^2-6\mathbb{I}\right)^5=0$.}
\begin{equation}
  \mathbf{J}^2
  \left(\mathbf{J}^2-2\mathbb{I}\right)
  \left(\mathbf{J}^2-6\mathbb{I}\right)=0\ ,
  \label{eq:spin1minimal}
\end{equation}
which implies that every higher power
of $\mathbf{J}^2$ reduces to a quadratic polynomial. Thus
\begin{equation}
  H=c_0\mathbb{I}+c_1\mathbf{J}^2
    +c_2\left(\mathbf{J}^2\right)^2.
\end{equation}
The coefficient $c_0$ produces only an overall phase in the
time-evolution operator, leaving two independent interactions.

With $S_k$ denoting the spin-$1$ generators, one of the rotationally invariant interactions can be chosen as the ``dipole'' interaction
\begin{equation}
  \sum_{k=1}^{3}S_k\otimes S_k
  =\frac12\left(\mathbf{J}^2-4\mathbb{I}\right).
  \label{eq:spin1bilinear1}
\end{equation}
The term quadratic in $\mathbf{J}^2$ is related to the traceless ``quadrupole'' interaction,
\begin{equation}
  Q_{k\ell}\equiv
  \frac12\{S_k,S_\ell\}
  -\frac23\delta_{k\ell}\mathbb{I},
\end{equation}
which is absent for spin $1/2$. More explicitly,
\begin{align}
  \sum_{k,\ell=1}^{3}Q_{k\ell}\otimes Q_{k\ell}
    &=
    \left(\sum_{k=1}^{3}S_k\otimes S_k\right)^2
    \nonumber\\
    &\quad+\frac12\sum_{k=1}^{3}S_k\otimes S_k
      -\frac43\mathbb{I}.
  \label{eq:spin1quadrupole}
\end{align}
It is customary to write the most general Hamiltonian in the 
 bilinear--biquadratic (BLBQ) form:
\begin{align}
  H_{\rm BLBQ}
    &=J\sum_{k=1}^{3}S_k\otimes S_k
%    \nonumber\\
   % &\quad
      +K\left(\sum_{k=1}^{3}S_k\otimes S_k\right)^2\ .
 %   \nonumber\\
   % &=\sum_{j=0}^{2}E_jP_j,
 % \nonumber\\
 % (E_0,E_1,E_2)
   % &=(-2J+4K,\,-J+K,\,J+K).
  \label{eq:blbqH}
\end{align}
Using 
\begin{align}
  \sum_{k=1}^{3}S_k\otimes S_k
    &=-2P_0-P_1+P_2,
%  \nonumber\\
 % \left(\sum_{k=1}^{3}S_k\otimes S_k\right)^2
   % &=4P_0+P_1+P_2.
  \label{eq:spin1bilinear}
\end{align}
it is easy to solve for the eigenvalues of $H_{\rm BLBQ}$:
\be 
(E_0,E_1,E_2)
   =(-2J+4K,\,-J+K,\,J+K).
\ee
The three energy eigenvalues  contain two independent differences,
which become the two relative phases of the corresponding
three-projector time-evolution operator.

\begin{figure*}[t]
  \centering
  \includegraphics[width=0.8\textwidth]{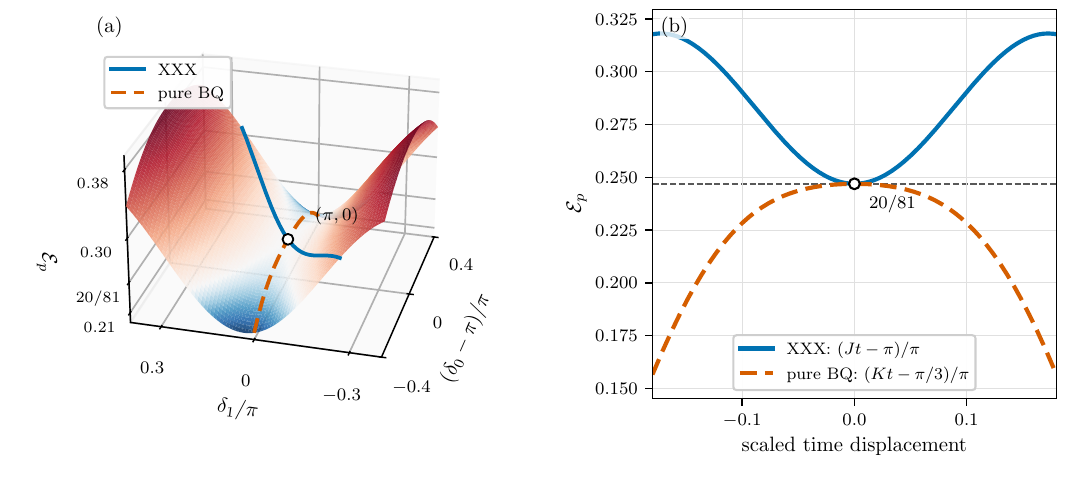}
  \caption{\label{fig:spin1_saddle}
  Entangling power near the stationary corner
  $(\delta_0,\delta_1)=(\pi,0)$ of the two-site spin-$1$ BLBQ family.
  (a) Entangling power in relative-phase coordinates, with the XXX
  (solid blue) and pure-biquadratic (dashed orange) trajectories.
  (b) Corresponding one-dimensional trajectories of the time-evolution: the corner  is a local minimum along the XXX trajectory
  but a local maximum along the pure-biquadratic trajectory.}
\end{figure*}

The full entangling-power  associated with the projectors
$P_{0,1,2}$ was obtained in the  context of high energy collisions in
Ref.~\cite{Kirchner:2023dvg}, while the instantaneous and
time-averaged dynamics of the two-site spin chain were studied in
Ref.~\cite{Low:2026oyf}. Here we use Theorem~\ref{thm:main} to
identify the stationary gates and determine how they appear along
physical Hamiltonian trajectories.

After removing the phase of the total spin $j=2$ channel, the
time-evolution operator  $U(t)=e^{iH_{\rm BLBQ}t}$ becomes
\begin{align}
  U(t)&\to V(\delta_0,\delta_1)
   =e^{i\delta_0}P_0+e^{i\delta_1}P_1+P_2,
  \nonumber\\
  \delta_0&%=(E_0-E_2)t
  =3(K-J)t,\quad 
  \delta_1=-2Jt.
  \label{eq:blbqmap}
\end{align}
The two relative phases depend independently on the bilinear and
biquadratic couplings, so varying $Jt$ and $Kt$ spans the full
two-dimensional  torus. For a fixed 
ratio $K/J$, however,  varying $t$ traces a one-dimensional
trajectory on the torus. Away from the energy-degeneracy ratios, all
three energies displayed above are distinct and the evolution
genuinely involves three projectors.

Recall that the even total-spin sectors $j=0,2$ are symmetric under
interchange of the two sites, whereas the $j=1$ sector is
antisymmetric. The SWAP gate is therefore
\begin{equation}
  S_W=P_0-P_1+P_2.
\end{equation}
Theorem~\ref{thm:main} identifies the four stationary corners:
\begin{align}
  V(0,0)&=\mathbb{I},
  &
  V(0,\pi)&=S_W,
  \nonumber\\
  V(\pi,0)&=\mathbb{I}-2P_0,
  &
  V(\pi,\pi)&=S_W(\mathbb{I}-2P_0).
  \label{eq:spin1corners}
\end{align}
The first two are the Identity gate and the SWAP gate and hence have
$\mathcal{E}_p=0$.

At the corner $(\pi,0)$, the reflected projector is $Q=P_0=\ketbra{s}{s}$, where
\begin{equation}
  \ket{s}
  =\frac{1}{\sqrt3}\left(
    \ket{1}\otimes\ket{-1}
    -\ket{0}\otimes\ket{0}
    +\ket{-1}\otimes\ket{1}
  \right)\ .
\end{equation}
Its reduced density matrices are
$\rho_A=\rho_B=\mathbb{I}/3$. Direct evaluation of the projector
invariants gives
\begin{align}
  &r=1,\qquad
  \mu_A=\mu_B=\frac13,\nonumber \\
  &\kappa=I_R(P_0)=I_T(P_0)=q_3=\frac19.
\end{align}
It follows from Eqs.~\eqref{eq:alpha} and \eqref{eq:beta} that
  $\alpha=8$ and $\beta={16}/{9}$. Since $D=3^2(3+1)^2=144$, the corner formula
Eq.~\eqref{eq:cornervalue} gives
\begin{equation}
  \mathcal{E}_p(\pi,0)
  =\frac{8}{144}\left(8-\frac{32}{9}\right)
  =\frac{20}{81}.
\end{equation}
Finally, $V(\pi,\pi)=S_WV(\pi,0)$. Although $S_W$ is not a local unitary, it only interchanges the bipartite state and, therefore, preserves the entanglement property. Therefore,
\begin{equation}
  \mathcal{E}_p(\pi,\pi)
  =\mathcal{E}_p(\pi,0)
  =\frac{20}{81}.
  \label{eq:spin1Ep}
\end{equation}

The two entangling corners are saddle points. Fig.~\ref{fig:spin1_saddle}
displays this structure at $(\pi,0)$, where two physical Hamiltonian
trajectories have opposite local curvatures. The isotropic Heisenberg,
or XXX, limit has
$K=0$,
\begin{equation}
  (\delta_0,\delta_1)=(-3Jt,-2Jt)\ ,
\end{equation}
and reaches $(\pi,0)$ when $Jt=\pi$. By contrast, in the pure
biquadratic limit $J=0$,
\begin{equation}
  (\delta_0,\delta_1)=(3Kt,0)\ ,
\end{equation}
and its time-evolution reaches the same corner when $Kt=\pi/3$. Along these two
trajectories,
\begin{align}
  \left.
  \frac{d^2\mathcal{E}_p}{dt^2}
  \right|_{K=0,\,Jt=\pi}
    &=\frac{31}{27}J^2>0,
  \nonumber\\
  \left.
  \frac{d^2\mathcal{E}_p}{dt^2}
  \right|_{J=0,\,Kt=\pi/3}
    &=-\frac{2}{9}K^2<0.
  \label{eq:blbqcurvatures}
\end{align}
Thus the same stationary corner is a local minimum along the XXX evolution
and a local maximum along the pure-biquadratic evolution.  The same conclusion holds at
$(\pi,\pi)$ by SWAP invariance. Theorem~\ref{thm:main} fixes the
stationarity, while the Hamiltonian trajectory determines which
one-dimensional extremum is observed.

Restricting Eq.~\eqref{eq:blbqmap} to the XXX trajectory reproduces
Eq.~(13) of Ref.~\cite{Low:2026oyf}. In particular,
Theorem~\ref{thm:main} explains without differentiation why the
nonzero-entanglement point $Jt=\pi$ is stationary. The time-averaged
dependence on the BLBQ coupling ratio studied in
Ref.~\cite{Low:2026oyf}, however, is controlled by spectral
degeneracies and is not fixed by the theorem.

Taken together, these examples separate universal corner stationarity from model-dependent classification: the projectors fix the phase landscape, while a Hamiltonian selects the trajectory through it.

\section{\label{sec:con} Discussion and Outlook}
In this work, we provided a general theorem identifying stationary points of the entangling power generated by unitary operators acting on a bipartite Hilbert space. In particular, given the spectral decomposition of a unitary operator with $n$ distinct
eigenvalues, we have shown that stationary points exist at the $2^{n-1}$ \textit{corners} of phase inversion,
where every relative eigenphase is $0$ or $\pi$. At each of these points, the unitary operator takes the form $\mathbb{I}-2Q$ up to an overall phase, where $Q$ is the sum of all spectral projectors whose relative phase is $\pi$. We also obtained the entangling power at every
corner in terms of invariants of $Q$. Further, the set of corner operators forms the group $\mathbb{Z}_2^{\,n-1}$. Conversely, we have shown that a
unitary can be realized as a corner of some projector family if and
only if $U^2\propto\mathbb{I}$, a criterion obeyed by both Clifford
and non-Clifford gates. In the special case that the spectral decomposition contains only two distinct eigenvalues, these observations allow us to calculate the entangling power over the entire relative-phase circle.

Notably, while Theorem~\ref{thm:main} guarantees the existence of stationary points, their classification as minima, maxima, or saddle points depends on the specific dynamics of the quantum system. Recall that the two-projector $SU(N)$ families discussed in Sec.~\ref{sec:twoproj} furnish non-entangling minima for the $\mathbf{N}\otimes\mathbf{N}$ family and, for $N\ge 3$, nonzero-entanglement maxima for the $\mathbf{N}\otimes\overline{\mathbf{N}}$ family. The spin-$1$ example in Sec.~\ref{sec:threeproj} further shows that the type of extremum observed along a physical time-evolution depends on the trajectory through the stationary point on the full phase torus. More specifically, in the two-site spin-$1$ BLBQ model, the same saddle is a local minimum along the XXX time-evolution trajectory but a local maximum along the pure-biquadratic trajectory.

The theorem also provides a useful perspective on recent studies of entanglement and enhanced symmetries in low-energy hadronic scattering~\cite{Low:2021ufv,Liu:2022grf,Liu:2023bnr,Hu:2024hex,Hu:2025lua,Low:2026evp}. In these systems, the $S$-matrix decomposes into projectors in spin and/or isospin space, $S=\sum_a e^{2i\delta_a}P_a$.
The enhanced-symmetry points of the two-nucleon $S$-matrix are the Identity gate and the SWAP gate, both of which have vanishing entangling power and coincide with the two-projector corners of Theorem~\ref{thm:main}. For the four-channel scattering of distinguishable spin-$3/2$ baryons studied in Ref.~\cite{Hu:2025lua}, the theorem identifies six additional stationary corners beyond the Identity gate and the SWAP gate. It would be interesting to determine whether any of these corners likewise exhibits an enhanced symmetry.

Several limitations should be kept in mind. The identical-particle construction of Ref.~\cite{Hu:2025lua}, for example, uses a projected sub-unitary $S$-matrix and a modified input average and therefore lies outside the assumptions of Theorem~\ref{thm:main}. Moreover, there are stationary points which are not corners; already for two channels, the maximum lies in the interior when $\alpha<4\beta$. For three or more projectors, classifying a corner generally requires the full Hessian and, when the Hessian is degenerate, higher-order terms. For a non-entangling unitary, the entangling power vanishes and therefore attains its global minimum. The theorem also identifies stationary corners with nonzero entangling power. Finally, we use linear entropy averaged over product inputs; other entanglement measures or input ensembles must be studied separately. 

A natural future direction is to ask which stationary corner gates can be synthesized from symmetry-preserving few-body operations. A basic example is the construction of a three-qudit gate from a sequence of symmetry-preserving two-qudit gates. This question connects Theorem~\ref{thm:main} with reachability results for symmetry-preserving quantum circuits~\cite{Marvian:2020sjz,Marvian:2021ysn,Marvian:2023vps}. For three qudits with $SU(N)$ symmetry, $N\geq3$, the fully symmetric and fully antisymmetric representations each occur once, whereas the mixed-symmetry representation occurs twice. A generic three-qudit $SU(N)$-invariant gate therefore has four distinct eigenvalues and hence four spectral projectors. After removing an overall phase, its entangling power is a function on a three-torus of relative eigenphases, on which Theorem~\ref{thm:main} identifies eight stationary corners. It would be interesting to apply the no-ancilla reachability criterion of Ref.~\cite{Hulse:2024ttn} to see which, if any, of the eight corners can be synthesized using only two-qudit $SU(N)$-invariant gates. It would also be interesting to compare the entangling powers of the reachable and excluded corners and determine whether they are minima, maxima, or saddles on the torus.

A complementary direction is to ask whether the reflection structure of corner gates has practical consequences for their implementation. The numerical circuit-synthesis methods of Refs.~\cite{Ashhab:2022wcx,Ashhab:2023dap} could be used to compare the resources required to synthesize a corner gate with those required for nearby nonreflection gates in the same projector family, using a common native gate set and fidelity threshold. On the experimental side, a platform with tunable controlled-phase gates, such as the transmon-qutrit platform of Ref.~\cite{Goss}, could be used to sweep the relative phases through a corner. Estimating the entangling power on both sides of the crossing would test the predicted vanishing slope and determine the curvature along the implemented trajectory. Such studies would connect the geometric results obtained here to quantum control and gate characterization.

\begin{acknowledgments}
We thank Sahel Ashhab for comments on a previous draft and Pallab Goswami for helpful discussions.
We acknowledge the use of Claude (Anthropic) and ChatGPT (OpenAI) for assistance with computations, writing, and figure generation. The authors take full responsibility for the correctness and scientific value of the work. The work of N.M. is supported in part by the U.S. Department of Energy under grant No. DEFG02-13ER41976/DE-SC0009913, and the DOE QuantISED program through the theory consortium “Intersections of QIS and Theoretical Particle Physics” at Fermilab (FNAL 20-17) under contract No. 89243024CSC000002. I.L. is supported
in part by the U.S. Department of Energy under contracts DE-AC02-06CH11357 (Argonne), DE-SC0023522 (Northwestern), DE-SC0010143 (Northwestern), and No. 89243024CSC000002 (QuantISED Program). 
\end{acknowledgments}

\bibliography{ref}% Produces the bibliography via BibTeX.

\end{document}